\PassOptionsToPackage{expansion=false}{microtype}
\documentclass[a4paper,UKenglish,cleveref,autoref,thm-restate]{lipics-v2021}

\usepackage{booktabs}
\usepackage{array}
\usepackage{xurl}
\usepackage[table]{xcolor}
\newcommand{\chg}[1]{#1}
\hideLIPIcs
\nolinenumbers

\title{More RLBWT Tricks}
\titlerunning{More RLBWT Tricks}

\author{Travis Gagie}{CeBiB -- Center for Biotechnology and Bioengineering, Chile}{travis.gagie@gmail.com}{https://orcid.org/0000-0003-3689-327X}{}
\authorrunning{T. Gagie}
\Copyright{Travis Gagie}

\ccsdesc[500]{Theory of computation~Data structures design and analysis}
\ccsdesc[500]{Theory of computation~Pattern matching}
\ccsdesc[300]{Applied computing~Bioinformatics}
\keywords{Run-length compressed Burrows--Wheeler transform, unitigs, shortcuts, two-level indexing, closed syncmers, maximal exact matches, pangenomics}

\supplement{\url{https://github.com/TravisGagie/more-rlbwt-tricks}}

\acknowledgements{Eddie Ferro noted that the parse at closed syncmers is prefix-free~\cite{FerroPC}.  The code, all the experiments and the tables and paragraphs presenting the experimental results were done by Claude, an AI model by Anthropic, under the author's direction.  The key ideas were the author's but Claude gave useful feedback.}

\EventEditors{John Q. Open and Joan R. Access}
\EventNoEds{2}
\EventLongTitle{42nd Conference on Very Important Topics (CVIT 2016)}
\EventShortTitle{CVIT 2016}
\EventAcronym{CVIT}
\EventYear{2016}
\EventDate{December 24--27, 2016}
\EventLocation{Little Whinging, United Kingdom}
\EventLogo{}
\SeriesVolume{42}
\ArticleNo{23}

\begin{document}
\maketitle

\begin{abstract}
Indexes based on $k$-mers can beat BWT-based indexes when the former take advantage of unitigs in de Bruijn graphs.  In this paper we consider ways in which BWT-based indexes can do that as well, or even better.  In particular, we consider MEM-finding between human chromosome 19s, where we can achieve a five- to sevenfold speedup over plain backward search, or a two- to threefold speedup over plain backward search with latency hiding, for about twice the space.  We also show that two-level indexing can be fast without an index of the text: without one, MEM-finding is seven to eight times faster than plain backward search in slightly less space, and from the parse alone we can then locate every occurrence of every MEM for about 5\% more time and 3 to 10\% more space.
\end{abstract}

\section{Introduction}
\label{sec:intro}

Several years ago, Ben Langmead noted that indexes based on $k$-mers can take advantage of long paths in de Bruijn graphs (dBGs), called \emph{unitigs}~\cite{Pibiri2022}, and asked whether indexes based on the Burrows--Wheeler transform (BWT) can benefit similarly.
Specifically, if an index based on $k$-mers that is searching for a pattern reaches the first vertex in a unitig of length $\ell$ whose label is readily accessible, then it can compare the next $\ell$ characters of the pattern to the unitig's label and either skip to the end of the unitig or declare a mismatch.
Since string comparison is faster than the backward stepping that BWT-based indexes use, this significantly speeds up many pattern-matching tasks when, for example, the indexed text is a collection of human haplotypes and the pattern is another human haplotype.
Unfortunately, indexes based on $k$-mers generally cannot by themselves find with certainty maximal exact matches (MEMs) of length more than $k + 1$ --- they must check candidates against the text, and the longer a MEM is relative to $k$, the more of the work that checking does --- and the MEMs of such patterns with respect to such texts are often huge --- exactly when skipping unitigs would pay off most!
Ben's question is therefore particularly relevant for bioinformatics applications, such as imputation, that can use such MEMs~\cite{FerroEtAl2026SIFT}.
It led some of us to develop two-level indexing~\cite{HongEtAl2024,BrownEtAlRLCSA} (Deng et al.~\cite{DengEtAl2022} had proposed something similar), suffixient sets~\cite{DepuydtEtAl2023,CenzatoEtAl2024} and suffix-tree path decompositions (STPDs)~\cite{BeckerEtAl2026}, which refine suffixient sets.
Ayad et al.'s U-index~\cite{AyadEtAl2025} also takes advantage of the speed of string comparison relative to backward stepping, but it answers only queries for all the occurrences of a whole pattern of at least a fixed length, with the text held in memory, so it does not report MEMs, and with 64 individuals our text alone would exceed our memory limit.

It occurred to us recently that we might be overthinking things, however.
In this paper we consider some new approaches to unitig-like skipping with BWT-based indexes, and give preliminary experimental results demonstrating their effectiveness.
In Section~\ref{sec:preliminaries} we review two practical MEM-finding algorithms, as well as two-level indexing and STPDs.
In Section~\ref{sec:funnels} we define a generalization of unitigs called funnels and explain how BWT-based indexes can take advantage of them.
In Section~\ref{sec:fun} we combine funnels with two-level indexing, and in Section~\ref{sec:encodings} we encode phrases that are unique successors.
In Section~\ref{sec:latency} we consider prefetching and latency hiding.
In Section~\ref{sec:textless} we show that two-level indexing does not always even need an index of the text, and that we can still locate the occurrences of the MEMs from the parse alone.

{
For our preliminary experiments, we took chromosome~19 from CHM13~\cite{NurkEtAl2022} and from both haplotypes of 66 individuals in release~2 of the Human Pangenome Reference Consortium~\cite{HPRC2025,LucasEtAl2026}: the first 66 individuals, by sample identifier, whose haplotypes each have chromosome~19 as a single sequence of at least 90\% of the length of CHM13's.
We found the MEMs of length at least 1\,000 of haplotype~1 of two of those individuals, HG00423 and HG00438 (split at runs of \texttt{N} into 4 patterns of 122.4 million characters in all) with respect to CHM13 and both haplotypes of the first 16, 32 or 64 of the other 64 individuals; there are 22\,711, 18\,479 and 16\,806 such MEMs.
\chg{With their reverse complements, the three texts have 4.07, 8.00 and 15.90 billion characters, and their BWTs have 68.4, 72.6 and 79.4 million runs.
All our sizes and times are in \Cref{tab:summary}, our only table, at the end of the paper.}
We used a desktop with an Intel Core i7-8700 (3.2\,GHz; 32\,KiB of L1 data cache and 256\,KiB of L2 cache per core, and 12\,MiB of shared L3 cache) and 16\,GB of RAM, running Ubuntu 26.04.1 (Linux 7.0), with all code compiled by GCC 15.2 at \texttt{-O3} with \texttt{-march=native}; every search ran on one thread pinned to one core, with memory limited to 11\,GB, and we report the median of three runs, excluding the time to load the index; in \Cref{tab:summary} the three runs of an entry differ by at most 3.2\%, and for most entries by less than 1\%.
Every search ran with transparent huge pages, which reduced the times with all the tools, by between 1\% with the STPD-index and 20\% with the RLCSA of the text.
Our code and scripts are available at \url{https://github.com/TravisGagie/more-rlbwt-tricks}.
\par}

\section{Preliminaries}
\label{sec:preliminaries}

\subsection{MEM-finding}

A maximal exact match (MEM) of a pattern $P [0..m - 1]$ with respect to a text $T [0..n - 1]$ is a substring $P [i..j]$ of $P$ such that
\begin{itemize}
\item $P [i..j]$ occurs in $T$,
\item either $i = 0$ or $P [i - 1..j]$ does not occur in $T$,
\item either $j = m - 1$ or $P [i..j + 1]$ does not occur in $T$.
\end{itemize}
Among bioinformaticians, the most popular way to find MEMs between patterns and an indexed text is probably Li's forward--backward algorithm~\cite{Li2012}, which he used in his tool BWA-MEM~\cite{Li2013}.
Tools such as copMEM2~\cite{GrabowskiBieniecki2023} find all the maximal matches between a pair of genomes by sampling $k$-mers from both, and target very similar genomes, but they keep the uncompressed reference in memory and report each pair of matching positions separately.
Against a collection of similar genomes, both their space and their output therefore grow with the collection, whereas a BWT-based index reports each MEM of the pattern once, with its BWT interval.

{
\par}

\subsubsection{Forward-backward}

Forward--backward uses an index for $T$ that
\begin{itemize}
\item given a suffix $P [k..m - 1]$ of $P$, returns the length of the longest prefix of $P [k..m - 1]$ that occurs in $T$,
\item given a prefix $P [0..k]$ of $P$, returns the length of the longest suffix of $P [0..k]$ that occurs in $T$.
\end{itemize}
For general texts we would normally use one FM-index or CSA for $T$ to support the longest-suffix queries and another for the reverse of $T$ to support the longest-prefix queries, but if $T$ is the concatenation of a genomic database and its reverse complement then we need only an FM-index or CSA for $T$.
Indexing both strands also means that we search each pattern once and find its MEMs with respect to both strands.

Assuming for simplicity that every character in $P$ occurs in $T$ --- otherwise we can break $P$ into maximal substrings including only characters that occur in $T$, and search each such substring separately --- the leftmost MEM in $P$ is the longest prefix $P [0..e_1]$ of $P$ that occurs in $T$.
If $P [s_q..e_q]$ is the $q$th MEM from the left and $e_q < m - 1$, then the $(q + 1)$st MEM from the left contains $P [e_q + 1]$, so the longest suffix $P [s_{q + 1}..e_q + 1]$ of $P [0..e_q + 1]$ that occurs in $T$ is a prefix of the $(q + 1)$st MEM from the left, and the longest prefix of $P [s_{q + 1}..m - 1]$ that occurs in $T$ is the whole $(q + 1)$st MEM from the left.
Thus, we can find all the MEMs from left to right.

\subsubsection{BML}

If we are interested only in MEMs of length at least $L$, called $L$-MEMs, then we can find them more quickly in practice with a recent variant of forward--backward called Boyer--Moore--Li (BML)~\cite{Gagie2024BML,Li2024}.
If $m < L$ then there are no $L$-MEMs.
Otherwise, let $P [s..L - 1]$ be the longest suffix of $P [0..L - 1]$ that occurs in $T$.
If $s > 0$ then no $L$-MEM starts in $P [0..s - 1]$ and we can continue searching in $P [s..m - 1]$.
If $s = 0$ then the leftmost $L$-MEM is the longest prefix $P [0..e]$ of $P$ that occurs in $T$.
If $e < m - 1$ and $P [s'..e + 1]$ is the longest suffix of $P [0..e + 1]$ that occurs in $T$, then no other $L$-MEM starts in $P [0..s' - 1]$ and we can continue searching in $P [s'..m - 1]$.

Forward--backward and BML with FM-indexes or CSAs also return the BWT intervals of the MEMs they find, meaning the lexicographic ranges of the suffixes of $T$ starting with occurrences of those MEMs.

\subsection{Two-level indexing}
\label{sec:twolevel}

For two-level indexing, we build a parse of the text that is prefix-free, an FM-index of the text based on its run-length compressed BWT (RLBWT), an index of the parse based on its BWT or RLBWT, and a sparse bitvector mapping suffixes of the text that start at phrase boundaries to the corresponding suffixes of the parse and back again.
(Readers may be familiar with Rsync-like PFP~\cite{BoucherEtAl2019}, but we used a newer version~\cite{FerroEtAl2026SIFT} that uses a parse based on closed syncmers~\cite{Edgar2021}: a $k$-mer is a closed syncmer if a smallest of its $s$-mers is its first or its last, in an order in which an $s$-mer and its reverse complement are equal; the phrases start at the closed syncmers.)
Given a pattern, we start matching right to left, character by character, in the index of the text until we reach the last phrase boundary in the pattern; jump up into the index of the parse; match phrase by phrase until we reach the leftmost complete phrase; jump back down into the index of the text; and finish by matching character by character again.
If the BWT interval becomes empty while we are matching in the index of the text, we have found the longest suffix of the pattern that occurs in the text.
If the BWT interval becomes empty while we are matching in the index of the parse, we jump down from the leftmost matching phrase into the index of the text and continue matching character by character until we reach a mismatch.
We can thus use two-level indexing with forward--backward and BML.
We identify the phrases of a pattern by looking up 64-bit hashes of them in a dictionary of the hashes of the phrases of the text, and in Section~\ref{sec:fun} we compare fingerprints of phrases' identifiers, so two-level search is exact with high probability, assuming the hash functions behave like random functions; in all our experiments its output is the same as plain backward search's.

{
For the index of the text we use either an RLCSA or a move structure~\cite{NishimotoTabei2021}~\cite{BrownEtAl2022}; \chg{\Cref{tab:summary} shows that}, in our implementations, plain backward search with our move structure is about twice as fast as with the RLCSA.
We timed ropebwt3 at commit b408339, four commits after release 3.10, and Movi~2 with its regular index and its table of 10-mers.
KeBaB~\cite{BrownEtAl2025} does not help with these patterns: with $k = 20$ its filter keeps 99.5\% of their characters and, in an earlier session with 16 individuals, ropebwt3 took 63.9 seconds after KeBaB and 64.3 seconds alone.
Our RLCSA stores two integers for each run of the BWT, the run's first row and the number of copies of the run's character before that row, in a list for each character; a table of jumps into each list narrows a rank query to a binary search among about 8 runs, and a backward step is two rank queries.
The two integers take 32 bits each with 16 individuals, when $n < 2^{32}$, and 40 bits each with 32 and 64 individuals.
(The two-level searches with the RLCSA of the text were timed with 64 bits each at 32 and 64 individuals; with 40 bits each those searches are no slower.)
Our move structure stores one 64-bit record for each run: the run's character (3 bits), the run's length (16 bits) and, for the LF-image of the run's first row, the run containing that image (29 bits) and the image's offset in that run (16 bits).
We split runs of more than 65\,535 rows, which adds at most 226 runs, as in length capping~\cite{BrownLangmead2026}, but we do not balance the runs~\cite{NishimotoTabei2021,BrownEtAl2026Construction}, so the walk over whole runs that finishes an LF step is not bounded in the worst case; here that walk crosses 1.2 to 1.3 runs on average.
A backward step with a character $c$ moves each end of the interval to the nearest run of $c$ inside the interval, by scanning up to 8 runs and otherwise by a query on an Elias--Fano encoding of the list of $c$'s runs, and then applies LF to both ends.
Those encodings and an Elias--Fano encoding of the runs' first rows, which converts between the move structure's positions and BWT rows, take 127 / 142 / 163\,MB of \chg{our move structure's sizes in \Cref{tab:summary}}, which also include a 2\,MB table of the intervals of the 8-mers, built when the index is loaded, as is common for FM-indexes~\cite{AndersonWheeler2021}.
Throughout, a size is the size of the files a search loads plus the tables built at load that we mention; the memory a search holds once it has loaded is 0.5 to 2.9\% more, mostly for the program itself and for select structures (and slightly less with the STPD of the parse), which changes no comparison between rows.
The searches without encoding were timed on files that also held the pairs of \Cref{sec:encodings}, which those searches never read; we report the sizes without the pairs.
\chg{We built everything on one thread with memory limited to 11\,GiB (11.8\,GB), and we give construction times in seconds and peak memory in GB; we have not tried to reduce either, and our builders keep inputs in memory that a careful builder would not.
With 16 / 32 / 64 individuals, writing the text from the FASTA files takes 19 / 36 / 69 seconds and computing its BWT with Big-BWT~\cite{BoucherEtAl2019} takes 370 / 627 / 1\,047 seconds and 2.7 / 3.6 / 5.0\,GB.
From the BWT, building the RLCSA takes 27 / 57 / 107 seconds and 1.3 / 1.5 / 1.6\,GB, and building our move structure takes about 10 seconds and 2.6 / 2.8 / 3.0\,GB.}
Because the fields of a record are relative, our move structure grows only with the number of runs.

Movi~2's move structure (we timed commit 9bc0acc) is nearly the same as ours: it also takes 8 bytes for each run and computes LF the same way, it splits runs of more than 4\,095 rows and it does not balance the runs; it is smaller because it has no counterpart of our Elias--Fano encodings.
BML is 1.5 to 1.7 times slower with Movi~2 than with our move structure because of how the two searches use their move structures, not because of the move structures themselves.
Movi~2 finds the right end of a MEM by maintaining a bidirectional interval, which costs it 4 to 5 times as many instructions per backward step as our search executes; our search finds the right end by backward search on the reverse complement of the pattern, which takes 2 to 5\% more backward steps.
Measured without huge pages, the two searches spend the same number of cycles stalled on cache misses.
\chg{\Cref{tab:summary} shows the times with BML only.
With our indexes that include an index of the text, forward--backward takes within 4\% of the time of BML, except with the STPD of the parse (4.45 / 4.55 / 4.77 seconds with 16 / 32 / 64 individuals); without the index of the text, forward--backward is 4 to 11\% slower than BML.
Forward--backward takes 60.4 / 66.1 / 73.8 seconds with ropebwt3 and 58.5 / 50.2 seconds with the STPD-index, and Movi~2's own query for all MEMs takes 55.2 / 62.4 / 81.8 seconds.}
\chg{\Cref{tab:summary} shows the times} with our two-level index, at closed syncmers with $k = 51$ and $s = 12$, with each index of the text, both strands\chg{; the parses have 220, 434 and 862 million phrases, of which 8.5, 9.5 and 11.0 million are distinct, and their BWTs have 10.3, 13.0 and 17.4 million runs}.
\chg{Building the parse and the two-level index takes 239 / 348 / 658 seconds and 3.6 / 5.3 / 8.4\,GB; our current builder computes the parse from the run-length encoding of the BWT of the text.}

\par}

\subsection{STPDs}

For indexing with an STPD, we store a data structure supporting fast random access to the text and, for each path in a path decomposition of the suffix tree of the text, the position in the text of the character at the top of the path, in the co-lexicographic order of the prefixes of the text that end at those characters.
Given a pattern, we start at the first character of the pattern and repeatedly
\begin{itemize}
\item choose one of the stored positions so as to maximize the length of the longest common suffix of the prefix of the text ending at that position and the prefix of the pattern ending at the next unmatched character of the pattern;
\item find the longest common prefix of the suffix of the text starting just after that position and the suffix of the pattern starting just after that character.
\end{itemize}
This gives us the right ends of left-maximal matches between the pattern and the text.  Using random access to find the left ends of those matches na\"ively could take time proportional to the total length of the matches, which could be quadratic in the length of the pattern.
However, suppose we were matching the pattern against an occurrence of a string $\alpha$, reached a mismatch, and continued against an occurrence of a string $\beta$; then $\beta$ starts to the right of $\alpha$ in the pattern, and the part of the pattern between the left end of $\beta$ and the mismatch has already been compared with the occurrence of $\alpha$.
So when we reach the left end of $\beta$ on the way back, we resume comparing with the occurrence of $\alpha$ from there, rather than from the mismatch.
The parts of the pattern we compare are then disjoint, so the number of character comparisons is linear in the length of the pattern plus the number of matches we find, which is also linear in the length of the pattern.

{
\chg{\Cref{tab:summary} shows the times} with Cenzato et al.'s STPD-index~\cite{STPDindexCode} as published.  With 64 individuals its construction algorithm needed more than our desktop's 16\,GB RAM.  We note that STPDs return one occurrence of each match rather than its BWT interval, although they can be augmented to support fast locating queries.
In preliminary experiments, an STPD implementation tuned for finding very long MEMs, without locating, was considerably faster than Cenzato et al.'s; engineering STPDs carefully for this task is future work.

\par}

\subsection{Two-level indexing with STPDs}

A curious reader may wonder whether we get any benefit by using a two-level index with an RLFM-index for the text and an STPD for the parse.
(Using an STPD for the text would not work well, since we could miss chances to jump up into the index for the parse.)
The existing implementation of STPDs~\cite{STPDindexCode} is specialized to DNA and cannot handle the integer alphabet of a parse, so we implemented our own to try this --- and it does seem to work.

{
\chg{\Cref{tab:summary} shows the results}, with the RLFM-index and the RLCSA of the text.
Two-level search with the STPD is 1.5 to 1.6 times faster with BML than with the RLCSA of the parse, but the whole index is 2.8 to 4.4 times larger, and it grows faster with the collection.
\chg{Building the STPD of the parse takes 157 / 309 / 644 seconds and 3.5 / 6.1 / 11.6\,GB, after 656 / 1\,210 / 2\,348 seconds and 5.9 / 7.2 / 9.4\,GB for our own STPD of the text, which this search loads for the left end of a match.}
As with STPDs of the text, the search finds one occurrence of each MEM rather than its BWT interval; if we need all the occurrences, we can start from that one and use $\phi$ and LCP structures to list the others, until the LCP drops below the length of the MEM.
Replacing the plain parse by relative Lempel--Ziv makes the STPD 50\% slower, and it still takes 1.5 times the space of the encodings in \Cref{sec:encodings}.

\par}

\section{Funnels}
\label{sec:funnels}

Consider a sequence of $\ell$ edges $(v_0, v_1), (v_1, v_2), \ldots, (v_{\ell - 1}, v_\ell)$ in the dBG such that, for $0 < i \leq \ell$, $v_i$ has in-degree 1 (so its only incoming edge is from $v_{i - 1}$).
We call such a sequence of edges a \emph{funnel} because, when we are backward stepping to find a pattern, if at some point the BWT interval for the suffix of the pattern that we have already processed is a subinterval of the BWT interval for the $k$-mer at $v_i$ then, unless the pattern ends or stops matching anything in the text, we will backward step to BWT intervals of the same length that are subintervals of the BWT intervals for the $k$-mers at $v_{i - 1}, v_{i - 2}, \ldots, v_0$.
In other words, we are funnelled towards $v_0$.
The name echoes funnel DAGs~\cite{MillaniEtAl2020,Caceres2023}, in which the vertices of in-degree 1 are called forking vertices; the dBG itself is generally not even acyclic, let alone a funnel in their sense.
In work on safe genome assembly, Cairo et al.~\cite{CairoEtAl2023} call a walk join-free when each of its arcs enters a vertex of in-degree at most 1.
Notice that a unitig is a funnel in which all the vertices except possibly the last also have out-degree 1.
If we store the starting position of the BWT interval for $v_\ell$'s $k$-mer, the starting position of the BWT interval for the funnel's whole label (which starts with the $k$-mer at $v_0$ and ends with the $k$-mer at $v_\ell$), and the funnel's label, then when the BWT interval for the suffix of the pattern we have processed is a subinterval of the BWT interval for $v_\ell$'s $k$-mer, we can quickly compare the preceding characters in the pattern against the funnel's label and, if they match, skip backward stepping through the funnel: the new BWT interval has the same length and the same offset in the BWT interval for the funnel's label as the old one has in the BWT interval for $v_\ell$'s $k$-mer.

If the dBG contains a sequence of edges $(v_0, v_1), (v_1, v_2), \ldots, (v_{\ell - 1}, v_\ell)$ whose label occurs very often in the text but that is not quite a funnel, because there are a few places in the text where the $k$-mer at a vertex $v_i$ in those edges is preceded by a character other than the first character of $v_{i - 1}$'s $k$-mer --- so $v_i$ has in-degree more than 1 in the dBG --- then, along with the information we would store with this \emph{almost-funnel} if it were a funnel, we can store a bitvector that lets us treat it like a funnel.
Specifically, that bitvector has a bit for each suffix of the text starting with $v_\ell$'s $k$-mer, with 0s marking suffixes that are preceded by the rest of the almost-funnel's label and 1s marking suffixes that are not.
If we are backward stepping to find a pattern and at some point the BWT interval for the suffix of the pattern that we have already processed is a subinterval of the BWT interval for the $k$-mer at $v_\ell$, then we can quickly compare the preceding characters in the pattern against the almost-funnel's label and, if they match, again skip backward stepping through the almost-funnel.
{The only complications are that this time, the BWT interval for the suffix of the pattern we reach
\begin{itemize}
\item may be shorter than the one we started from,
\item may have a different offset in the BWT interval for the almost-funnel's label than the one we started from has in the BWT interval for $v_\ell$'s $k$-mer.
\end{itemize}
We can compute its length and offset with rank queries on the bitvector: its rows are the 0s in the part of the bitvector for the interval we started from, so its offset is the number of 0s before that part, and its length is the number of 0s in that part.
If it is empty, all of the occurrences we were following have left the almost-funnel, and we go back to backward stepping from where we were.\par}

Readers familiar with Baier's tunnels in RLBWTs~\cite{Baier2018} may have noticed that funnels are something like tunnels.
More precisely, a funnel is what Baier et al.~\cite{BaierEtAl2022} call a prefix interval, starting at the BWT interval of a $k$-mer; they showed that the maximal paths of fusible edges in a dBG, which are its unitigs, correspond to disjoint prefix intervals, whereas we ask only for in-degree 1.
Following Baier, we can generalize funnels and almost-funnels to start at BWT intervals that do not correspond to $k$-mers.
As with tunnels~\cite{BaierDede2019}, there are obvious optimization questions about funnels and almost-funnels.  We leave those for future work and, for now, we chose chromosome 19 of CHM13 as a reference and cut it into windows of 25 characters and, for each window that occurs only once in the reference, built an almost-funnel ending at that window over the 25 characters before it, so that a search crosses one window per step; since the almost-funnels all have the same length and chain, a search that has crossed one window checks for the next one.
{We built these almost-funnels along both CHM13 and its reverse complement, dropping windows that occur more than once on either strand, on our move structure, which indexes both strands.
\chg{\Cref{tab:summary} shows that} with windows of 25 characters they make BML 1.4 times faster than plain backward search on the same move structure, for 1.2 to 1.3 times its space, and with windows of 100 characters they make it 2.5 to 2.6 times faster, for about 1.1 times its space.
\chg{Building the almost-funnels over windows of 100 characters takes 34 / 32 / 36 seconds and 1.6 / 2.3 / 3.6\,GB.}
Windows of 25 characters need four times as many jumps as windows of 100 to cross the same distance, and each jump costs about a microsecond, mostly to convert between the move structure's positions and BWT rows.\par}

\section{More fun with funnels}
\label{sec:fun}

As with STPDs, we can easily mix funnelling with two-level indexing.
Another obvious extension is allowing variable-length almost-funnels.
The almost-funnels we have considered so far help us only if the pattern matches from one end of the almost-funnel to the other, but if we use wavelet trees instead of bitvectors then we can chain a sequence of short almost-funnels together to take advantage of matches that last one short almost-funnel or several.
To test these ideas, we adapted our indexes from the previous section to allow
\begin{itemize}
\item skipping almost-funnels with the index of the parse in two-level indexing,
\item skipping $2^t$ almost-funnels for $t \leq 6$ with the index of the text.
\end{itemize}

In the parse, an almost-funnel crosses phrases instead of characters: every $g$th phrase of the reference's parse ends an almost-funnel over the $g$ phrases before it, and to check the pattern against its label we store a Karp--Rabin fingerprint of those phrases' identifiers.

In the text, we store almost-funnels of length $25 \cdot 2^i$ only every $25 \cdot 2^i$ characters.
Since the positions of the 0s in the bitvector indicating which suffixes of the text match $2^i$ short almost-funnels are a superset of the positions of the 0s in the bitvector indicating which suffixes match $2^{i + 1}$, we can combine the bitvectors stored for position $25 \cdot 2^j$ into a wavelet tree over an alphabet of size $j + 2$, indicating whether each suffix matches $0, 1, 2, 4, \ldots,$ or $2^j$ short almost-funnels.

{
\chg{The results are in \Cref{tab:summary}.}  On the text, variable-length almost-funnels make BML 3.2 to 3.3 times faster than plain backward search on the same move structure, against 1.4 times and 2.5 to 2.6 times for fixed-length almost-funnels over windows of 25 and 100 characters, for 1.4 to 1.8 times the space of the move structure, mostly for the small numbers: new haplotypes add suffixes starting with the windows without adding almost-funnels.
\chg{Building these almost-funnels takes 99 / 88 / 104 seconds and 3.7 / 6.1 / 10.8\,GB.}
On the parse, with our move structure for the text, almost-funnels make two-level search 9 to 12\% faster with $g = 4$ or $16$, for 1 to 2\% more space\chg{; building them with $g = 16$ takes 26 / 48 / 99 seconds and 2.3 / 4.1 / 7.7\,GB}.

\par}

While we used a linear reference in our experiments in the last section, it would also be natural to use a graph reference~\cite{BaaijensEtAl2022}.  However, that would raise even more questions about optimization~\cite{AlankoEtAl2019,GogaBalaz2022}, especially since tunnels and funnels have the nice feature that, if a BWT-based index reports whenever it follows them, then we get a summary of the path a pattern-matching search took, rather than just the BWT interval we end up with.
When dealing with pangenome graphs this could be useful, but then we should consult geneticists about the choice of tunnels and funnels, so that they are biologically meaningful and worth reporting.

As a preliminary step in this direction, however, we considered the parse as a directed graph with phrases as vertices and an edge from each phrase $v$ to each $u$ that precedes it in the dataset.
From each vertex, we consider the path that always takes the edge that keeps the most suffixes matching, and end the path when following one more edge would mean fewer than a certain percentage of suffixes would still match.

{
We built an almost-funnel over each such path, with the percentage $\theta$ set to 80\%, 90\% or 95\% and paths of at most 64 phrases, and checked the pattern against them with fingerprints, as above.
\chg{\Cref{tab:summary} shows that} these almost-funnels make two-level search 1.9 to 2.3 times faster, and 1.7 to 2.0 times faster than the almost-funnels from the reference's parse, for 13 to 17\% more space in all; about 40\% of their space is fingerprints, so shorter fingerprints would shrink them.
\chg{Building them with $\theta = 80\%$ takes 137 / 171 / 275 seconds and 2.6 / 4.5 / 8.3\,GB.}
The choice of $\theta$ matters little.
With 16 individuals and $\theta = 90\%$, searches take an almost-funnel 714\,738 times, and in 38\,149 of them some of the matches drop out midway; on average an almost-funnel crosses 24.6 phrases.

\par}

\section{Encodings}
\label{sec:encodings}

Considering the paths in the graph of the phrases and Ben's question about unitigs that started us on this line of research, it occurred to us that there might be a more elegant way to skip unitigs (and more).
Suppose that, if a phrase $\alpha$ is always followed by the same phrase $\beta$ in the dataset, then we replace copies of $\beta$ immediately after copies of $\alpha$ in the parse by a special symbol $\triangleright$.
As long as we store which phrases are always followed by which other phrases, we can easily decode the encoded parse.
Notice that $\beta$ need not always be preceded by $\alpha$, so sequences of phrases we replace by $\triangleright$s are like funnels in the reverse of the text.

Now consider the encoding as a sequence of meta-characters and its suffixes starting at non-$\triangleright$ meta-characters.
Since a $\triangleright$ that follows a given meta-character always encodes the same phrase, those suffixes have the same lexicographic order encoded as they would decoded, assuming the lexicographic order of phrases' meta-characters is the same as that of the phrases themselves.
To see why, consider two suffixes and how we can determine their relative lexicographic order without decoding them.
If the first mismatch is in position $i$ then either one of the suffixes ends at position $i - 1$ (and this cannot happen if the text ends with a special symbol) or both have non-$\triangleright$ meta-characters in position $i$.

Conceptually, we can further encode the parse by combining each non-$\triangleright$ meta-character $\alpha$ and the (possibly empty) run of $j$ copies of $\triangleright$ after it into a pair $(\alpha, j)$.
We can build an FM-index for the sequence of pairs and backward search in it: at the right end of the pattern, the simplest solution --- and the one we have implemented so far --- is to match character by character in the index of the text and then phrase by phrase in the index of the parse, as in two-level indexing, until we reach a phrase of the pattern that starts a pair, and then restrict the interval to the suffixes of the parse that start pairs; symmetrically, when we reach the first pair in the pairs-encoding of the pattern, the simplest solution --- and the one we have implemented so far --- is to drop down and match phrase by phrase again.
A phrase can be encoded as $\triangleright$ after one phrase and start a pair after another, so the interval after a step in the index of the pairs holds only the occurrences in which the pair's first phrase starts a pair. The interval after the last step in the index of the pairs is exact when the first phrase of the last pair crossed is never encoded as $\triangleright$. Otherwise we go back to the interval before that last step, which is exact for the occurrences preceded by the pair, and match the pair's phrases one by one from there.
Jumping over pairs is something like adding a third level of indexing, in which the pairs are the phrases of a parse of the parse; as with two-level indexing, we work at the coarsest level in the middle of a match and drop down a level near its ends.
In the full version of this paper we will explain how to avoid this dropping down, which can be inefficient, and how we can encode $\beta$ when it is the most frequent phrase that follows $\alpha$ even when it is not the only one.
(We jump up to the pairs sooner at the right end in Section~\ref{sec:textless}.)
For now, our results with encodings are \chg{in \Cref{tab:summary}}.
\chg{Building the encoding takes 59 / 124 / 279 seconds and 2.2 / 4.2 / 8.2\,GB.}

{
With 16, 32 and 64 individuals, 85.1\%, 82.6\% and 79.4\% of the phrases in the parse are encoded as $\triangleright$, so the sequence of pairs has 32.8, 75.4 and 177.3 million pairs, each encoding 6.7, 5.8 and 4.9 phrases on average.
With our move structure for the text, encoding makes two-level search 1.7 to 1.9 times faster, for 4 to 6\% more space.

\par}

\section{Prefetching and latency hiding}
\label{sec:latency}

{
Much of the time of all the searches above is spent waiting for memory --- Sanaullah et al.~\cite{SanaullahEtAl2026} recently called memory throughput a fundamental speed limit for RLBWT-based methods --- so we tried three ways to wait less.
\chg{The results are in \Cref{tab:summary}.}

\subparagraph*{A table of first runs.}
In our RLCSA of the parse, a backward step for a phrase $c$ started with a predecessor search to find the first run of $c$ in the BWT.
We now store, for each phrase, its first run and that run's position in F, in 8 bytes, and likewise for each pair in the index of the pairs of \Cref{sec:encodings}.
The tables are built when the index is loaded and take about 78, 90 and 109\,MB with 16, 32 and 64 individuals, of which 68, 76 and 88\,MB are for the phrases.
They make plain two-level search 11\% faster, two-level search with heavy paths 4 to 5\% faster, and two-level search with encoding 5 to 6\% faster.

\subparagraph*{Lookahead.}
Prefetching is standard for searches in FM-indexes~\cite{VasimuddinEtAl2019,AndersonWheeler2021,GrootKoerkamp2026}.
In two-level indexing the phrases of the pattern are known before the search starts, so while the search is at one phrase it can prefetch what the steps for later phrases will read: the entry in the table of first runs for the phrase four steps ahead and, using that entry, the part of the RLCSA at that phrase's first run for the phrase two steps ahead.
Most phrases have few runs (86 to 87\% of the steps are for phrases whose runs fit in one cache line), so this covers most steps completely.
With the table, lookahead makes plain two-level search another 21 to 24\% faster, so 1.4 to 1.5 times faster in all; without the table it gains only 4 to 6\%, because the search for the first run cannot be prefetched this way.
The distance hardly matters: looking one, two or four steps ahead gave the same times to within 1\%.
With encoding, the two changes together make the search 7 to 10\% faster; there the lookahead covers only the index of the parse, not the index of the pairs.

\subparagraph*{Latency hiding.}
Movi~\cite{ZakeriEtAl2024} prefetches in its query loop, and Movi~2~\cite{ZakeriEtAl2026} hides the latency of cache misses by interleaving searches on one thread: before a search accesses the index, the search prefetches the address and yields to another search.
Interleaving independent lookups on one thread to hide cache misses comes from in-memory databases~\cite{ChenEtAl2007,KocberberEtAl2015}, where the lookups are now written as coroutines~\cite{PsaropoulosEtAl2017,JonathanEtAl2018,KirianskyEtAl2018}.
Our patterns are few and long, so we cut them into 64 pieces that are searched independently and interleave $K$ pieces at a time in coroutines.
Each piece is searched as if the pattern started at the piece's left end, without stopping at the piece's right end, and reports the MEMs that start in the piece; a match that starts exactly at a cut is discarded if the last MEM reported before the cut has the same right end, since then it is not left-maximal.
The output is the same as without cutting, and cutting costs under 1\% by itself.
In each backward step on our move structure, a piece prefetches the two rows the step leads to and yields; a switch between pieces takes 11 to 14 nanoseconds.
Plain backward search spends 70\% of its cycles stalled on memory (cycles with a pending memory load and no execution, as counted by \texttt{perf stat} in user space, with 16 individuals and huge pages), and with $K = 8$, the best of the values from 4 to 32 that we tried, it is 2.5 times faster and, with 16 individuals, spends 32\% of its cycles stalled on memory.

We did not do the same for two-level indexing, because our measurements say it would not significantly improve two-level indexing's performance.
In profiles taken without huge pages, stalls on cache misses in the move structure account for between 5\% and 7\% of the cycles of two-level search, and in the RLCSA of the parse, with the table and lookahead in place, stalls on the cache misses that interleaving could hide and the lookahead cannot account for at most 2.0\%.

\par}

\section{Searching without the index of the text}
\label{sec:textless}

Two-level search uses the index of the text only near the ends of a match, and we can do without that index altogether.
Deng et al.~\cite{DengEtAl2022} built an RLBWT over a grammar-based parse with no index of the text, but resolving the left end of a pattern cost them a backward step for each of many candidate phrases, and their index counted occurrences faster than an RLFM-index only for patterns of hundreds to thousands of characters.
Hong et al.~\cite{HongEtAl2024} observed that a pattern that starts and ends at phrase boundaries can be searched in the index of the parse alone.
Here we resolve both ends of a match by string comparisons with phrases chosen by the search itself.
Consecutive phrases overlap by $k$ characters, since each phrase ends with the closed syncmer with which the next phrase starts.
Call the part of the pattern to the right of its last phrase boundary the tail, and let $\alpha$ be the pattern's last complete phrase: the tail starts with the last $k$ characters of $\alpha$, and the phrases that start with those $k$ characters have consecutive identifiers, because we number the phrases in lexicographic order.
If we store each phrase without its first $k$ characters, then a binary search among those phrases, by string comparison with the rest of the tail, gives the range of phrases that start with the tail, which is an interval of suffixes of the parse.
At the left end, when the next phrase of the pattern does not match, the match ends inside the phrase to the left of the current phrase boundary: we list the distinct phrases that precede the suffixes of the parse in the current interval and compare the pattern leftward with each of those phrases without its last $k$ characters; the longest agreement gives the left end of the match, and the suffixes preceded by a phrase with the longest agreement are its occurrences. In 85 to 88\% of our searches with BML there is only one such phrase.
So in place of the index of the text and the bitvector we store two copies of each phrase, one without its first $k$ characters and one without its last $k$ characters, at 2 bits per character; for each phrase, the range of identifiers of the phrases that start with its last $k$ characters; and the heads of the runs of the BWT of the parse, in the order of its rows.
With encoding we also jump up to the pairs at the right end: if $\alpha$ has no unique successor then $\alpha$ ends its pair at every occurrence, so we restrict the interval for the tail to the suffixes that start pairs and take a step in the index of the pairs at once; and if the unique successor of $\alpha$ starts with the tail, which takes one comparison to check, then we start from the whole interval of the pair that contains $\alpha$ in the pattern.

A match must contain closed syncmers for this search to find it.
(Durbin's syng~\cite{Durbin2026} also searches a BWT over closed syncmers, a GBWT with $k = 63$ and $s = 8$, with no index of the text, but it ends matches at syncmers and so, as he notes, finds MEMs only approximately; our searches are exact.)
Among any $k - s$ consecutive $k$-mers at least one is a closed syncmer~\cite{Edgar2021}, so with $k = 51$ and $s = 12$ every string of at least 128 characters over the DNA alphabet contains two phrase boundaries, and $L = 1\,000$ is far above that.
Searching this way gives up three things.
First, when no phrase starts with the tail we learn only that the match ends at or to the right of the last phrase boundary.
That is enough for BML's tests of whether a window occurs, which only make BML skip a little less far, and the searches that find the ends of a MEM run on strings that are known to occur.
Second, forward--backward needs the exact ends of matches that may contain no closed syncmer, so we changed it: when a search returns only a bound we restart from the bound, and we report a candidate only if its length is at least $L$ and it does not occur with one more character on its left, which we can check exactly.
Third, like STPDs, searching without the index of the text does not return the interval of a MEM in the BWT of the text.
The STPD-index as published~\cite{STPDindexCode} returns one occurrence of each MEM, although STPDs can support locating all the occurrences~\cite{BeckerEtAl2026}; our search by itself returns the number of occurrences of each MEM but no occurrence.
We can locate from the parse alone, however.
The search that finds the right end of a MEM leaves the parse with an interval in the suffix array of the parse that contains a row for each occurrence of the MEM.
To count the occurrences, that search already compares the characters of the MEM to the left of its first phrase with the phrases preceding the rows in that interval, run by run, so it can keep the rows that agree.
We store the position in the text of every 64th phrase, marking the rows of those phrases in a sparse bitvector as in standard suffix-array sampling~\cite{MakinenEtAl2010}, and, from each row we keep, we step forward in the parse, from the row of a suffix of the parse to the row of the next suffix, until we reach a stored position, adding the lengths of the phrases we pass.
We walk the rows of all the MEMs of a pattern together, prefetching as in Section~\ref{sec:latency}, and rows that are consecutive in the same run walk as one, as when Sir\'en's RLCSA locates a range of rows at once~\cite{Siren2012}.
The stored positions and the bitvector take 18 / 36 / 74 MB with 16 / 32 / 64 individuals.
\chg{Building them takes 52 / 83 / 162 seconds and 2.4 / 4.4 / 8.3\,GB.}
Locating every occurrence of every MEM adds 0.10 to 0.14 seconds to BML with almost-funnels along heavy paths and the table of first runs, which then takes 2.38 to 2.46 seconds, and at most 0.19 seconds to the other searches we timed with locating.
With the index of the text, keeping a toehold through the search and listing the occurrences with $\phi^{-1}$ adds 1.3 seconds to BML and 889 / 1\,181 MB with 16 / 32 individuals.
Our move structure alone returns only the BWT interval of a MEM; listing the occurrences from that interval needs the samples and $\phi^{-1}$ of the text, which are 624 / 687 MB of those 889 / 1\,181 MB.
The BWT of the parse has only 21 / 34 / 49 phrases per run, so storing positions at the ends of runs as the $r$-index does~\cite{GagieEtAl2020} would take more space here; subsampling~\cite{CobasEtAl2025} bounds the space either way.
Every position we report is an occurrence of the MEM in the text, and each MEM gets as many distinct positions as it has occurrences.
In all our experiments the MEMs found and their numbers of occurrences are the same as with plain backward search.

\chg{The results are in the last group of \Cref{tab:summary}.
Building the text-free index from the parse and the two-level index takes about 10 seconds and 1.0 / 1.2 / 1.5\,GB.}
Plain two-level search without the index of the text takes 0.47 to 0.59 of the space of our move structure alone, and with BML it is 3.1 to 3.2 times faster than plain backward search and 4 to 10\% faster than two-level search with the index of the text.
With almost-funnels along heavy paths and the table of first runs, BML takes 2.25 to 2.32 seconds: 1.2 times faster than the same search with the index of the text, in 0.43 to 0.49 of its space, and 7.1 to 8.2 times faster than plain backward search, in 0.79 to 0.97 of the space of our move structure; it is also 2.8 to 3.3 times faster than plain backward search with latency hiding.
The indexes without the index of the text grow faster with the collection than our move structure does, however, so that advantage in space shrinks as we add individuals.
The jump at the right end makes search with encoding 3 to 5\% faster.

\section{Summary}
\label{sec:summary}

Table~\ref{tab:summary} summarizes our results with BML, with entries on the Pareto frontier shaded.
The take-away message is that, with two-level indexing, (reversed) funnels and encoding, we can speed up MEM-finding between human chromosomes by a factor of 4.8 to 7.1, compared to plain backward search with the same move structure, and by a factor of 1.9 to 2.8 compared to plain backward search with latency hiding, for 1.6 to 2.0 times the space.
Without the index of the text, two-level search with almost-funnels along heavy paths is 7.1 to 8.2 times faster than plain backward search with the same move structure, and 2.8 to 3.3 times faster than plain backward search with latency hiding, in 0.79 to 0.97 of the space; it is 10 to 14 times faster than BML with Movi~2, the fastest of the other tools.

\chg{We also timed the searches in \Cref{tab:summary} with thresholds of 250 and 5\,000 instead of 1\,000, for which there are 32\,979 / 26\,031 / 23\,244 and 9\,885 / 9\,747 / 9\,523 MEMs.
The times of our searches hardly change with the threshold, although the fastest searches are faster at 5\,000: plain backward search with our move structure takes 15.7 / 17.5 / 18.8 seconds at 250 and 16.3 / 18.2 / 19.4 seconds at 5\,000, and two-level search with almost-funnels along heavy paths ($\theta = 80\%$) and the table takes 2.84 / 2.71 / 2.74 and 2.28 / 2.33 / 2.38 seconds, so it is 5.5 to 6.9 times faster at 250 and 7.1 to 8.2 times faster at 5\,000.
The other tools vary more: BML with Movi~2 takes 26.0 / 31.0 / 38.1 seconds at 250 and 19.4 / 23.0 / 26.1 seconds at 5\,000, the STPD-index takes 46.9 / 49.3 and 22.3 / 30.2 seconds, and ropebwt3 takes 60.8 / 68.4 / 74.3 and 54.1 / 66.2 / 72.9 seconds.
Of the searches without the index of the text, we timed only plain two-level search at these thresholds; it takes 4.95 / 5.49 / 6.27 and 5.09 / 5.63 / 6.33 seconds.}

With our current builders, everything that the fastest search without the index of the text needs takes 16 / 26 / 47 minutes from the FASTA files with 16 / 32 / 64 individuals; 7 / 11 / 19 of those minutes are for the BWT of the text and its run-length encoding, from which our builder computes the parse, although the search itself does not use them.
Building the samples and $\phi^{-1}$ of the text for locating with the two-level index exceeded our memory limit with 64 individuals.

\begin{table}[!htb]

\centering
\caption{Summary of our experiments: finding the MEMs of length at least 1\,000 of two chromosome~19 haplotypes with BML, with 16, 32 and 64 individuals; time in seconds and size in MB. All the indexes index both strands. The two-level indexes use our move structure for the text, except where we say otherwise and for the STPD of the parse ($^*$), which uses the RLCSA, and their sizes include the indexes of the text, except in the last group (which has none). ``Funnels'' means almost-funnels, and ``table'' means the table of first runs, whose size is included. The STPD-index and the STPD of the parse return one occurrence of each MEM and the indexes in the last group return its number of occurrences, not its BWT interval. Shaded entries are on the Pareto frontier of time and space with that many individuals.}
\label{tab:summary}
\footnotesize
\setlength{\tabcolsep}{4.5pt}
\begin{tabular}{@{}llrrrrrr@{}}
\toprule
 & & \multicolumn{3}{c}{seconds} & \multicolumn{3}{c}{MB} \\
\cmidrule(lr){3-5} \cmidrule(l){6-8}
index & method & 16 & 32 & 64 & 16 & 32 & 64 \\
\midrule
ropebwt3 & plain & \cellcolor{gray!25}59.94 & \cellcolor{gray!25}67.72 & \cellcolor{gray!25}73.80 & \cellcolor{gray!25}165 & \cellcolor{gray!25}187 & \cellcolor{gray!25}222 \\
Movi~2 & plain & 23.26 & 27.13 & 32.07 & 581 & 614 & 670 \\
STPD-index & plain & \cellcolor{gray!25}38.64 & \cellcolor{gray!25}43.58 & --- & \cellcolor{gray!25}264 & \cellcolor{gray!25}314 & --- \\
RLCSA & plain & 33.30 & 38.21 & 41.70 & 579 & 757 & 825 \\
\midrule
our move structure & plain & 15.99 & 17.70 & 18.98 & 668 & 716 & 791 \\
 & latency hiding, $K = 8$ & 6.31 & 6.99 & 7.58 & 668 & 716 & 791 \\
 & funnels, windows of 25 & 11.27 & 12.28 & 13.10 & 801 & 896 & 1\,064 \\
 & funnels, windows of 100 & 6.52 & 7.03 & 7.44 & 723 & 784 & 885 \\
 & funnels, $t \le 6$ & 5.00 & 5.37 & 5.70 & 912 & 1\,098 & 1\,448 \\
\midrule
two-level (RLCSA) & plain & 5.96 & 6.35 & 6.97 & 937 & 1\,169 & 1\,317 \\
\midrule
two-level (move) & plain & 5.53 & 5.87 & 6.48 & 1\,026 & 1\,128 & 1\,283 \\
 & table & 4.94 & 5.20 & 5.74 & 1\,094 & 1\,204 & 1\,371 \\
 & table and lookahead & 3.91 & 4.08 & 4.39 & 1\,094 & 1\,204 & 1\,371 \\
 & lookahead alone & 5.33 & 5.54 & 6.17 & 1\,026 & 1\,128 & 1\,283 \\
 & STPD of the parse$^*$ & 4.01 & 4.21 & 4.49 & 2\,618 & 3\,782 & 5\,747 \\
 & funnels, $g = 64$ & 5.34 & 5.62 & 6.19 & 1\,029 & 1\,132 & 1\,288 \\
 & funnels, $g = 16$ & 5.02 & 5.33 & 5.78 & 1\,033 & 1\,137 & 1\,295 \\
 & funnels, $g = 4$ & 4.99 & 5.25 & 5.72 & 1\,044 & 1\,150 & 1\,312 \\
 & heavy paths, $\theta = 95\%$ & 2.91 & 2.85 & 2.94 & 1\,157 & 1\,275 & 1\,452 \\
 & \quad with the table & 2.77 & 2.74 & 2.81 & 1\,225 & 1\,351 & 1\,540 \\
 & heavy paths, $\theta = 90\%$ & 2.86 & 2.80 & 2.86 & 1\,165 & 1\,283 & 1\,466 \\
 & \quad with the table & 2.75 & 2.68 & 2.75 & 1\,233 & 1\,360 & 1\,554 \\
 & heavy paths, $\theta = 80\%$ & 2.87 & 2.77 & 2.80 & 1\,175 & 1\,300 & 1\,495 \\
 & \quad with the table & 2.75 & 2.65 & 2.69 & 1\,243 & 1\,377 & 1\,583 \\
 & $\triangleright$ encoding & 3.31 & 3.30 & 3.50 & 1\,070 & 1\,186 & 1\,365 \\
 & \quad with the tables & 3.11 & 3.12 & 3.33 & 1\,149 & 1\,277 & 1\,474 \\
 & \quad with the tables and lookahead & 2.97 & 3.02 & 3.26 & 1\,149 & 1\,277 & 1\,474 \\
\midrule
two-level, no text & plain & \cellcolor{gray!25}5.00 & \cellcolor{gray!25}5.49 & \cellcolor{gray!25}6.20 & \cellcolor{gray!25}313 & \cellcolor{gray!25}369 & \cellcolor{gray!25}470 \\
 & table and lookahead & 3.40 & 3.63 & 4.00 & 381 & 446 & 558 \\
 & $\triangleright$ encoding & 2.93 & 3.02 & 3.28 & 357 & 428 & 553 \\
 & \quad with the jump & \cellcolor{gray!25}2.77 & \cellcolor{gray!25}2.92 & \cellcolor{gray!25}3.18 & \cellcolor{gray!25}357 & \cellcolor{gray!25}428 & \cellcolor{gray!25}553 \\
 & \quad\quad and the tables and lookahead & \cellcolor{gray!25}2.50 & \cellcolor{gray!25}2.67 & \cellcolor{gray!25}2.94 & \cellcolor{gray!25}436 & \cellcolor{gray!25}519 & \cellcolor{gray!25}662 \\
 & heavy paths, $\theta = 80\%$ & \cellcolor{gray!25}2.38 & \cellcolor{gray!25}2.38 & \cellcolor{gray!25}2.45 & \cellcolor{gray!25}462 & \cellcolor{gray!25}542 & \cellcolor{gray!25}682 \\
 & \quad with the table & \cellcolor{gray!25}2.26 & \cellcolor{gray!25}2.25 & \cellcolor{gray!25}2.32 & \cellcolor{gray!25}530 & \cellcolor{gray!25}618 & \cellcolor{gray!25}770 \\
\bottomrule
\end{tabular}
\end{table}


\clearpage
\begin{thebibliography}{10}

\bibitem{AlankoEtAl2019}
Jarno~N. Alanko, Travis Gagie, Gonzalo Navarro, and Louisa Seelbach~Benkner.
\newblock Tunneling on {Wheeler} graphs.
\newblock In {\em Proceedings of the Data Compression Conference (DCC)}, pages
  122--131, 2019.

\bibitem{AndersonWheeler2021}
Tim Anderson and Travis~J. Wheeler.
\newblock An optimized {FM}-index library for nucleotide and amino acid search.
\newblock {\em Algorithms for Molecular Biology}, 16:25, 2021.

\bibitem{AyadEtAl2025}
Lorraine A.~K. Ayad, Gabriele Fici, Ragnar Groot~Koerkamp, Grigorios Loukides,
  Rob Patro, Giulio~Ermanno Pibiri, and Solon~P. Pissis.
\newblock {U-Index}: A universal indexing framework for matching long patterns.
\newblock In {\em Proceedings of the 23rd International Symposium on
  Experimental Algorithms (SEA)}, LIPIcs, pages 4:1--4:18, 2025.
\newblock \href {https://doi.org/10.4230/LIPIcs.SEA.2025.4}
  {\path{doi:10.4230/LIPIcs.SEA.2025.4}}.

\bibitem{BaaijensEtAl2022}
Jasmijn~A. Baaijens, Paola Bonizzoni, Christina Boucher, Gianluca Della~Vedova,
  Yuri Pirola, Raffaella Rizzi, and Jouni Sir{\'e}n.
\newblock Computational graph pangenomics: a tutorial on data structures and
  their applications.
\newblock {\em Natural Computing}, 21(1):81--108, 2022.

\bibitem{Baier2018}
Uwe Baier.
\newblock On undetected redundancy in the {Burrows--Wheeler} transform.
\newblock In {\em Proceedings of the 29th Annual Symposium on Combinatorial
  Pattern Matching (CPM)}, volume 105 of {\em LIPIcs}, pages 3:1--3:15, 2018.
\newblock \href {https://doi.org/10.4230/LIPIcs.CPM.2018.3}
  {\path{doi:10.4230/LIPIcs.CPM.2018.3}}.

\bibitem{BaierDede2019}
Uwe Baier and Kadir Dede.
\newblock {BWT} tunnel planning is hard but manageable.
\newblock In {\em Proceedings of the Data Compression Conference (DCC)}, pages
  142--151, 2019.

\bibitem{BaierEtAl2022}
Uwe Baier, Thomas B{\"u}chler, Enno Ohlebusch, and Pascal Weber.
\newblock Edge minimization in de {B}ruijn graphs.
\newblock {\em Information and Computation}, 285:104795, 2022.

\bibitem{BeckerEtAl2026}
Ruben Becker, Davide Cenzato, Travis Gagie, Ragnar Groot~Koerkamp, Sung-Hwan
  Kim, Giovanni Manzini, and Nicola Prezza.
\newblock Compressing suffix trees by path decompositions.
\newblock In {\em Proceedings of the 53rd International Colloquium on Automata,
  Languages, and Programming (ICALP)}, volume 374 of {\em LIPIcs}, pages
  24:1--24:25, 2026.
\newblock \href {https://doi.org/10.4230/LIPIcs.ICALP.2026.24}
  {\path{doi:10.4230/LIPIcs.ICALP.2026.24}}.

\bibitem{BoucherEtAl2019}
Christina Boucher, Travis Gagie, Alan Kuhnle, Ben Langmead, Giovanni Manzini,
  and Taher Mun.
\newblock Prefix-free parsing for building big {BWT}s.
\newblock {\em Algorithms for Molecular Biology}, 14:13, 2019.

\bibitem{BrownEtAl2025}
Nathaniel~K. Brown, Lore Depuydt, Mohsen Zakeri, Anas Alhadi, Nour Allam, Dove Begleiter, Nithin Bharathi Kabilan~Karpagavalli, Suchith~Sridhar Khajjayam, Hamza Wahed, Travis Gagie, and Ben Langmead.
\newblock {KeBaB}: $k$-mer based breaking for finding long {MEMs}.
\newblock In {\em Proceedings of the 32nd International Symposium on String Processing and Information Retrieval (SPIRE 2025)}, volume 16073 of {\em LNCS}, pages 10--17, 2026.

\bibitem{BrownEtAl2022}
Nathaniel~K. Brown, Travis Gagie, and Massimiliano Rossi.
\newblock {RLBWT} tricks.
\newblock In {\em Proceedings of the 20th International Symposium on Experimental Algorithms (SEA)}, volume 233 of {\em LIPIcs}, pages 16:1--16:16, 2022.

\bibitem{BrownLangmead2026}
Nathaniel~K. Brown and Ben Langmead.
\newblock Bounding the average move structure query for faster and smaller {RLBWT} permutations.
\newblock In {\em Proceedings of the 24th Symposium on Experimental Algorithms (SEA)}, volume 371 of {\em LIPIcs}, pages 9:1--9:15, 2026.

\bibitem{BrownEtAlRLCSA}
Nathaniel~K. Brown, Travis Gagie, Giovanni Manzini, Gonzalo Navarro, and
  Marinella Sciortino.
\newblock Faster run-length compressed suffix arrays.
\newblock In {\em From Strings to Graphs, and Back Again: A Festschrift for
  Roberto Grossi's 60th Birthday}, OASIcs, pages 10:1--10:15, 2026.
\newblock \href {https://doi.org/10.4230/OASIcs.Grossi.10}
  {\path{doi:10.4230/OASIcs.Grossi.10}}.

\bibitem{BrownEtAl2026Construction}
Nathaniel~K. Brown, Ahsan Sanaullah, Shaojie Zhang, and Ben Langmead.
\newblock Optimal-time move structure construction.
\newblock Technical Report 2603.22147, arXiv, 2026.

\bibitem{Caceres2023}
Manuel C{\'a}ceres.
\newblock Parameterized algorithms for string matching to {DAGs}: Funnels and beyond.
\newblock In {\em Proceedings of the 34th Symposium on Combinatorial Pattern Matching (CPM)}, volume 259 of {\em LIPIcs}, pages 7:1--7:19, 2023.

\bibitem{CairoEtAl2023}
Massimo Cairo, Romeo Rizzi, Alexandru~I. Tomescu, and Elia~C. Zirondelli.
\newblock Genome assembly, from practice to theory: Safe, complete and linear-time.
\newblock {\em ACM Transactions on Algorithms}, 20(1):4, 2023.

\bibitem{CenzatoEtAl2024}
Davide Cenzato, Lore Depuydt, Travis Gagie, Sung-Hwan Kim, Giovanni Manzini,
  Francisco Olivares, and Nicola Prezza.
\newblock Suffixient arrays: A new efficient suffix array compression
  technique.
\newblock {\em Theory of Computing Systems}, 70:50, 2026.
\newblock \href {https://doi.org/10.1007/s00224-026-10287-6}
  {\path{doi:10.1007/s00224-026-10287-6}}.

\bibitem{ChenEtAl2007}
Shimin Chen, Anastassia Ailamaki, Phillip~B. Gibbons, and Todd~C. Mowry.
\newblock Improving hash join performance through prefetching.
\newblock {\em ACM Transactions on Database Systems}, 32(3):17, 2007.

\bibitem{STPDindexCode}
Davide Cenzato et~al.
\newblock {STPD-index}.
\newblock \url{https://github.com/regindex/STPD-index}, commit 8a432a4, 2026.

\bibitem{CobasEtAl2025}
Dustin Cobas, Travis Gagie, and Gonzalo Navarro.
\newblock Fast and small subsampled $r$-indexes.
\newblock {\em ACM Transactions on Algorithms}, 22(1):7, 2025.

\bibitem{DengEtAl2022}
Jin-Jie Deng, Wing-Kai Hon, Dominik K{\"o}ppl, and Kunihiko Sadakane.
\newblock {FM}-indexing grammars induced by suffix sorting for long patterns.
\newblock In {\em Proceedings of the Data Compression Conference (DCC)}, pages 63--72, 2022.

\bibitem{DepuydtEtAl2023}
Lore Depuydt, Travis Gagie, Ben Langmead, Giovanni Manzini, and Nicola Prezza.
\newblock Suffixient sets.
\newblock arXiv:2312.01359, 2023.

\bibitem{Durbin2026}
Richard Durbin.
\newblock A run-length-compressed skiplist data structure for dynamic {GBWTs} supports time and space efficient pangenome operations over syncmers.
\newblock Technical Report 10.64898/2026.03.26.714584, bioRxiv, 2026.

\bibitem{Edgar2021}
Robert Edgar.
\newblock Syncmers are more sensitive than minimizers for selecting conserved
  $k$-mers in biological sequences.
\newblock {\em PeerJ}, 9:e10805, 2021.
\newblock \href {https://doi.org/10.7717/peerj.10805}
  {\path{doi:10.7717/peerj.10805}}.

\bibitem{FerroPC}
Eddie Ferro.
\newblock Personal communication, 2026.

\bibitem{FerroEtAl2026SIFT}
Eddie Ferro, Rahul Varki, Mao-Jan Lin, Travis Gagie, Ben Langmead, and
  Christina Boucher.
\newblock {SIFT}: Syncmer-informed selection of haplotypes for {T2T}
  imputation.
\newblock bioRxiv preprint, 2026.
\newblock \href {https://doi.org/10.64898/2026.11.09.765394}
  {\path{doi:10.64898/2026.11.09.765394}}.

\bibitem{GagieEtAl2020}
Travis Gagie, Gonzalo Navarro, and Nicola Prezza.
\newblock Fully functional suffix trees and optimal text searching in {BWT}-runs bounded space.
\newblock {\em Journal of the ACM}, 67(1):2, 2020.

\bibitem{Gagie2024BML}
Travis Gagie.
\newblock How to find long maximal exact matches and ignore short ones.
\newblock In {\em Proceedings of the 28th International Conference on
  Developments in Language Theory (DLT)}, 2024.

\bibitem{GogaBalaz2022}
Adri{\'a}n Goga and Andrej Bal{\'a}{\v{z}}.
\newblock Prefix-free parsing for building large tunnelled {Wheeler} graphs.
\newblock In {\em Proceedings of the 22nd International Workshop on Algorithms
  in Bioinformatics (WABI)}, volume 242 of {\em LIPIcs}, pages 18:1--18:12,
  2022.
\newblock \href {https://doi.org/10.4230/LIPIcs.WABI.2022.18}
  {\path{doi:10.4230/LIPIcs.WABI.2022.18}}.

\bibitem{GrabowskiBieniecki2023}
Szymon Grabowski and Wojciech Bieniecki.
\newblock {copMEM2}: Robust and scalable maximum exact match finding.
\newblock {\em Bioinformatics}, 39(5):btad313, 2023.

\bibitem{GrootKoerkamp2026}
Ragnar Groot~Koerkamp.
\newblock {QuadRank}: Engineering a high throughput rank.
\newblock In {\em Proceedings of the 24th Symposium on Experimental Algorithms (SEA)}, volume 371 of {\em LIPIcs}, pages 20:1--20:23, 2026.

\bibitem{HongEtAl2024}
Aaron Hong, Marco Oliva, Dominik K{\"o}ppl, Hideo Bannai, Christina Boucher,
  and Travis Gagie.
\newblock {PFP-FM}: An accelerated {FM}-index.
\newblock {\em Algorithms for Molecular Biology}, 19:15, 2024.
\newblock \href {https://doi.org/10.1186/s13015-024-00260-8}
  {\path{doi:10.1186/s13015-024-00260-8}}.

\bibitem{HPRC2025}
{Human Pangenome Reference Consortium}.
\newblock {HPRC} data release 2.
\newblock \url{https://humanpangenome.org/hprc-data-release-2/}, 2025.

\bibitem{JonathanEtAl2018}
Christopher Jonathan, Umar~Farooq Minhas, James Hunter, Justin Levandoski, and Gor Nishanov.
\newblock Exploiting coroutines to attack the ``killer nanoseconds''.
\newblock {\em Proceedings of the VLDB Endowment}, 11(11):1702--1714, 2018.

\bibitem{KirianskyEtAl2018}
Vladimir Kiriansky, Haoran Xu, Martin Rinard, and Saman Amarasinghe.
\newblock Cimple: Instruction and memory level parallelism: A {DSL} for uncovering {ILP} and {MLP}.
\newblock In {\em Proceedings of the 27th International Conference on Parallel Architectures and Compilation Techniques (PACT)}, pages 30:1--30:16, 2018.
\newblock \href {https://doi.org/10.1145/3243176.3243185} {\path{doi:10.1145/3243176.3243185}}.

\bibitem{KocberberEtAl2015}
Onur Kocberber, Babak Falsafi, and Boris Grot.
\newblock Asynchronous memory access chaining.
\newblock {\em Proceedings of the VLDB Endowment}, 9(4):252--263, 2015.

\bibitem{Li2012}
Heng Li.
\newblock Exploring single-sample {SNP} and {INDEL} calling with whole-genome
  de novo assembly.
\newblock {\em Bioinformatics}, 28(14):1838--1844, 2012.

\bibitem{Li2013}
Heng Li.
\newblock Aligning sequence reads, clone sequences and assembly contigs with
  {BWA-MEM}.
\newblock arXiv:1303.3997, 2013.

\bibitem{Li2024}
Heng Li.
\newblock {BWT} construction and search at the terabase scale.
\newblock {\em Bioinformatics}, 40(12):btae717, 2024.

\bibitem{LucasEtAl2026}
Julian~K. Lucas, Prajna Hebbar, Wen-Wei Liao, Juan~F. Macias-Velasco, Adam~M.
  Novak, et~al.
\newblock {HPRC2}: A human pangenome reference with near-complete coverage of
  common genetic variation.
\newblock bioRxiv, \url{https://doi.org/10.64898/2026.07.21.739710}, 2026.

\bibitem{MakinenEtAl2010}
Veli M{\"a}kinen, Gonzalo Navarro, Jouni Sir{\'e}n, and Niko V{\"a}lim{\"a}ki.
\newblock Storage and retrieval of highly repetitive sequence collections.
\newblock {\em Journal of Computational Biology}, 17(3):281--308, 2010.

\bibitem{MillaniEtAl2020}
Marcelo~Garlet Millani, Hendrik Molter, Rolf Niedermeier, and Manuel Sorge.
\newblock Efficient algorithms for measuring the funnel-likeness of {DAGs}.
\newblock {\em Journal of Combinatorial Optimization}, 39:216--245, 2020.

\bibitem{NishimotoTabei2021}
Takaaki Nishimoto and Yasuo Tabei.
\newblock Optimal-time queries on {BWT}-runs compressed indexes.
\newblock In {\em Proceedings of the 48th International Colloquium on Automata,
  Languages, and Programming (ICALP)}, volume 198 of {\em LIPIcs}, pages
  101:1--101:15, 2021.

\bibitem{NurkEtAl2022}
Sergey Nurk et~al.
\newblock The complete sequence of a human genome.
\newblock {\em Science}, 376(6588):44--53, 2022.

\bibitem{Pibiri2022}
Giulio~Ermanno Pibiri.
\newblock Sparse and skew hashing of $k$-mers.
\newblock {\em Bioinformatics}, 38(Supplement 1):i185--i194, 2022.

\bibitem{PsaropoulosEtAl2017}
Georgios Psaropoulos, Thomas Legler, Norman May, and Anastasia Ailamaki.
\newblock Interleaving with coroutines: A practical approach for robust index joins.
\newblock {\em Proceedings of the VLDB Endowment}, 11(2):230--242, 2017.

\bibitem{SanaullahEtAl2026}
Ahsan Sanaullah, Degui Zhi, Heng Li, and Shaojie Zhang.
\newblock Detection of maximal exact matches in terabases of genome assemblies in compressed space and sublinear time.
\newblock Technical Report 10.64898/2026.11.10.765459, bioRxiv, 2026.

\bibitem{Siren2012}
Jouni Sir{\'e}n.
\newblock {\em Compressed Full-Text Indexes for Highly Repetitive Collections}.
\newblock PhD thesis, University of Helsinki, 2012.
\newblock Report A-2012-5.

\bibitem{VasimuddinEtAl2019}
Md Vasimuddin, Sanchit Misra, Heng Li, and Srinivas Aluru.
\newblock Efficient architecture-aware acceleration of {BWA-MEM} for multicore
  systems.
\newblock In {\em Proceedings of the IEEE International Parallel and Distributed
  Processing Symposium (IPDPS)}, pages 314--324, 2019.
\newblock \href {https://doi.org/10.1109/IPDPS.2019.00041}
  {\path{doi:10.1109/IPDPS.2019.00041}}.

\bibitem{ZakeriEtAl2024}
Mohsen Zakeri, Nathaniel~K. Brown, Omar~Y. Ahmed, Travis Gagie, and Ben Langmead.
\newblock Movi: A fast and cache-efficient full-text pangenome index.
\newblock {\em iScience}, 27(12):111464, 2024.

\bibitem{ZakeriEtAl2026}
Mohsen Zakeri, Nathaniel~K. Brown, Travis Gagie, and Ben Langmead.
\newblock Movi 2: Fast and space-efficient queries on pangenomes.
\newblock {\em Bioinformatics}, 42(7):btag362, 2026.
\newblock \href {https://doi.org/10.1093/bioinformatics/btag362}
  {\path{doi:10.1093/bioinformatics/btag362}}.

\end{thebibliography}
\end{document}